\documentclass[11pt,a4paper]{article}

\usepackage{amssymb, amsmath, tikz}
\usepackage[margin=2cm]{geometry}

\newtheorem{definition}{Definition}
\newtheorem{proposition}{Proposition}
\newtheorem{theorem}{Theorem}
\newtheorem{example}{Example}
\newtheorem{lemma}{Lemma}
\newtheorem{corollary}{Corollary}

\newcommand{\cl}{\mathrm{cl}}

\newcommand{\SE}{\mathrm{SE}}

\newcommand{\ork}{\mathrm{or}}
\newcommand{\irk}{\mathrm{ir}}
\newcommand{\urk}{\mathrm{ur}}
\newcommand{\lrk}{\mathrm{lr}}

\newcommand{\spancl}{\mathrm{span}}
\newcommand{\uspancl}{\mathrm{uspan}}

\newenvironment{proof}{\textbf{Proof}}{$\Box$\\}

\begin{document}
\title{Entropy of Closure Operators}
\author{Maximilien Gadouleau}

\maketitle

\begin{abstract}
The entropy of a closure operator has been recently proposed for the study of network coding and secret sharing. In this paper, we study closure operators in relation to their entropy. We first introduce four different kinds of rank functions for a given closure operator, which determine bounds on the entropy of that operator. This yields new axioms for matroids based on their closure operators. We also determine necessary conditions for a large class of closure operators to be solvable. We then define the Shannon entropy of a closure operator, and use it to prove that the set of closure entropies is dense. Finally, we justify why we focus on the solvability of closure operators only.
\end{abstract}

AMS 2010 Subject classification: 94A17, 06A15, 05B35.

\section{Introduction}

Network coding is a novel means to transmit data through a network, where intermediate nodes are allowed to combine the packets they receive \cite{ACLY00}. In particular, linear network coding \cite{LYC03} is optimal in the case of one source; however, it is not the case for multiple sources and destinations \cite{Rii04, DFZ05}. Although for large dynamic networks, good heuristics such as random linear network coding \cite{KM03, HMK+06} can be used, maximizing the amount of information that can be transmitted over a static network is fundamental but very hard in practice. Solving this problem by brute force, i.e. considering all possible operations at all nodes, is computationally prohibitive. The network coding solvability problem is given as follows: given a network (with corresponding graph, sources, destinations, and messages), can all the messages be transmitted? This problem is very difficult, for instance some networks correspond to determining whether $k$ mutually orthogonal Latin squares of order $A$ exist.

Several major advances have been made on this problem. First of all, it can always be reduced to a multiple unicast instance, where each source sends a different message, requested to a corresponding unique destination.  In \cite{Rii06}, the network coding solvability problem is reduced to a problem on arbitrary directed graphs, thus removing the asymmetry between sources, intermediate nodes, and destinations. Notably, \cite{Rii07} introduces the entropy of a directed graph (not to be mistaken with K\"orner's graph entropy in \cite{Kor71}); calculating this entropy solves the network solvability problem. A more combinatorial approach is then given by the so-called guessing number of a graph, which is closely related to the entropy \cite{Rii07}. The guessing number of graphs is studied further in \cite{GR11}, where it is proved that the guessing number of a directed graph is equal to the independence number of a related undirected graph. The guessing number of undirected graphs is further explored in \cite{CM11}.

A closure operator on the vertex set of a digraph is introduced in \cite{Gad13}. Network coding solvability is then proved to be a special case of a more general problem, called the closure solvability problem, for the closure operator defined on a digraph related to the network coding instance. The latter problem also generalises the search for ideal secret sharing schemes \cite{BD91}. The main interest of closure solvability is that it allows us to use closure operators which do not arise from digraphs (notably the uniform matroids) but which have been proved to be solvable over many alphabets. In this paper, we introduce the concept of the entropy of an arbitrary closure operator. Again, calculating the entropy of a closure operator determines whether this closure operator is solvable or not. Therefore, this paper aims at studying this quantity in detail. 

The rest of the paper is organised as follows. We review the closure operator associated to a digraph and the general closure solvability problem in Section \ref{sec:preliminaries}. In Section \ref{sec:ranks}, we introduce four kinds of rank functions for a given closure operator. This not only helps us derive bounds on the entropy, but we are also able to provide axioms for matroids that are, up to the author's knowledge, new. Section \ref{sec:Shannon} then studies a natural upper bound on the entropy, based on polymatroids. This helps us prove that the set of closure entropies contains all rational numbers above 1. Finally, Section \ref{sec:beyond} investigates the solvability problem beyond closure operators.

\section{Preliminaries} \label{sec:preliminaries}

\subsection{Closure operators} \label{sec:closure}

Throughout this paper, $V$ is a set of $n$ elements. A closure operator on $V$ is a mapping $\cl: 2^V \to 2^V$ which satisfies the following properties \cite[Chapter IV]{Bir48}. For any $X,Y \subseteq V$,
\begin{enumerate}
    \item \label{it:Xincl} $X \subseteq \cl(X)$ (extensive);

    \item \label{it:cl_increasing} if $X \subseteq Y$, then $\cl(X) \subseteq \cl(Y)$ (isotone);

    \item \label{it:cl(cl)} $\cl(\cl(X)) = \cl(X)$ (idempotent).
\end{enumerate}
A {\em closed set} is a set equal to its closure. For instance, in a group one may define the closure of a set as the subgroup generated by the elements of the set; the family of closed sets is simply the family of all subgroups of the group. Another example is given by linear spaces, where the closure of a set of vectors is the subspace they span. Closure operators are central in lattice theory and in universal algebra; moreover, any Galois connection is equivalent to a closure operator.

%The closure satisfies the following properties. For any $X,Y \subseteq V$,
%\begin{enumerate}
%    \item \label{it:cl=bigcap} $\cl(X)$ is equal to the intersection of all closed sets containing $X$;
%
%    \item \label{it:cl(cap)} $\cl(\cl(X) \cap \cl(Y)) = \cl(X) \cap \cl(Y)$, i.e. the family of closed sets is closed under intersection;
%
%    \item \label{it:cl(cup)} $\cl(X \cup Y) = \cl(\cl(X) \cup \cl(Y))$.
%
%    \item \label{it:axiom} $X \subseteq \cl(Y)$ if and only if $\cl(X) \subseteq \cl(Y).$
%\end{enumerate}

We refer to
$$
    r:= \min\{|b| : \cl(b) = V\}
$$
as the {\em rank} of the closure operator. Any set $b \subseteq V$ of size $r$ and whose closure is $V$ is referred to as a {\em basis} of $\cl$. There is a natural partial order on closure operators of the same set. We denote $\cl_1 \le \cl_2$ if for all $X \subseteq V$, $\cl_1(X) \subseteq \cl_2(X)$; then $r(\cl_1) \ge r(\cl_2)$.

We shall focus on two families of closure operators. Firstly, a {\em matroid} is a closure operator satisfying the Steinitz-Mac Lane exchange property\footnote{In order simplify notation, we shall identify a singleton $\{v\}$ with its element $v$}: if $X \subseteq V$, $v \in V$, and $u \in \cl(X \cup v) \backslash \cl(X)$, then $v \in \cl(X \cup u)$ \cite{Oxl06}. In particular, the uniform matroid $U_{r,n}$ of rank $r$ over $n$ vertices is defined by
$$
	U_{r,n}(X) = \begin{cases}
	V & \mbox{if } |X| \ge r\\
	X & \mbox{otherwise.}
	\end{cases}
$$
Secondly, let $D = (V,E)$ be a digraph on $n$ vertices (possibly with loops, but without any repeated arcs). The $D$-closure of any $X \subseteq V$ is given by $\cl_D(X) := X \cup Y$, where $Y$ is the largest acyclic set of vertices such that $Y^- \subseteq X \cup Y$ \cite{Gad13}. The $D$-closure is well defined, see \cite{Gad13} for an alternative definition. Recall that a {\em feedbak vertex set} is a set of vertices $X$ such that $V \backslash X$ induces an acyclic subgraph. The rank of $\cl_D$ is therefore the minimum size of a feedback vertex set of $D$.

\begin{example} \label{ex:clD}
The $D$-closure of some classes of graphs can be readily determined.
\begin{enumerate}
	\item If $D$ is acyclic, then $\cl_D = U_{0,n}$.
	
	\item If $D = C_n$ , the directed cycle, then $\cl_D = U_{1,n}$.
	
	\item If $D = K_n$, the complete graph, then $\cl_D = U_{n-1,n}$.
	
	\item If $D$ has a loop on each vertex, then $\cl_D = U_{n,n}$.
\end{enumerate}
Conversely, no other uniform matroid can be viewed as a $D$-closure.
\end{example}

\subsection{Partitions} \label{sec:functions}

A partition of a finite set $B$ is a collection of subsets, called parts, which are pairwise disjoint and whose union is the whole of $B$. We denote the parts of a partition $f$ as $P_i(f)$ for all $i$. If any part of $f$ is contained in a unique part of $g$, we say $f$ {\em refines} $g$. The equality partition $E_B$ with $|B|$ parts refines any other partition, while the universal partition (with only one part) is refined by any other partition of $B$. The {\em common refinement} of two partitions $f$, $g$ of $B$ is given by $h:= f \vee g$ with parts
$$
    P_{i,j}(h) = P_i(f) \cap P_j(g): P_i(f) \cap P_j(g) \ne \emptyset.
$$

We shall usually consider a tuple of $n$ partitions $f= (f_1,\ldots,f_n)$ of the same set assigned to elements of a finite set $V$ with $n$ elements. In that case, for any $X \subseteq V$, we denote the common refinement of all $f_v, v \in X$ as $f_X := \bigvee_{v \in X} f_v$. For any $X,Y \subseteq V$ we then have $f_{X \cup Y} = f_X \vee f_Y$.

\subsection{Closure solvability and entropy}

We now review the closure solvability problem \cite{Gad13}. The instance of the problem consists of a closure operator $\cl$ on $V$ with rank $r$, and a finite set $A$ with $|A| \ge 2$, referred to as the {\em alphabet}.

\begin{definition}
A {\em coding function} for $\cl$ over $A$ is a family $f$ of $n$ partitions of $A^r$ into at most $|A|$ parts such that $f_X = f_{\cl(X)}$ for all $X \subseteq V$.
\end{definition}

The problem is to determine whether there exists a coding function for $\cl$ over $A$ such that $f_V$ has $A^r$ parts. That is, we want to determine whether there exists an $n$-tuple $f = (f_1,\ldots,f_n)$ of partitions of $A^r$ in at most $|A|$ parts such that
\begin{align*}
	f_X &= f_{\cl(X)} \quad \mbox{for all } X \subseteq V,\\
	f_V &= E_{A^r}.
\end{align*}

We make several remarks concerning the closure solvability problem.
\begin{enumerate}
	\item Closures associated to digraphs are particularly relevant for network coding. Indeed, a network coding instance is solvable if and only if $\cl_D$ is solvable for some digraph $D$ related to the network coding instance \cite{Gad13}.

	\item When reduced to matroids, this is the problem of representation by partitions in \cite{Mat99}, which is equivalent to determining secret-sharing matroids \cite{BD91}.

	\item The solvability problem could be defined as searching for families of partitions of any set $B$ with $|B| \ge |A|^r$ such that $f_V = E_B$. However, this can only occur if $|B| = |A|^r$; moreover, since only the cardinality of $B$ matters, we can assume without loss that $B = A^r$.

	\item A coding function $f$ naturally yields a closure operator $\cl_f$ on $V$, where $\cl_f(X) = \{v \in V: f_{X \cup v} = f_X\} =  \bigcup \{Y : f_Y = f_X\}$; we then have $\cl \le \cl_f$. Therefore, if $\cl_2$ is solvable, then any $\cl_1$ with the same rank and $\cl_1 \le \cl_2$ is also solvable.
\end{enumerate}

For any partition $g$ of $A^r$, we define its entropy as
$$
	H(g) := r - |A|^{-r} \sum_i |P_i(g)| \log_{|A|} |P_i(g)|.
$$
The equality partition on $A^r$ is the only partition with full entropy $r$. Denoting $H_f(X) := H(f_X)$, we can recast the conditions above as
\begin{align*}
	H_f(v) &\le 1 \quad \mbox{for all } v \in V,\\
	H_f(X) &= H_f(\cl(X)) \quad \mbox{for all } X \subseteq V,\\
	H_f(V) &= r.
\end{align*}
The first two conditions are equivalent to $f$ being a coding function. In general, the rank cannot always be attained, hence we define the entropy of a closure operator $\cl$ over $A$ as the maximum entropy of any coding function for it:
$$
	H(\cl,A) := \max \{H_f(V) : f \mbox{ coding function for } \cl \mbox{ over } A\}.
$$
The entropy of $\cl$ is defined to be the supremum of all $H(\cl,A)$.

\section{Rank functions of closure operators} \label{sec:ranks}

In this section, we investigate the properties of closure operators in general and we derive bounds on the entropy of their coding functions. We shall introduce four kinds of ranks for any closure operator. It is worth noting that they are all distinct from the so-called rank function of a closure operator studied in \cite{Bat84}.

\subsection{Inner and outer ranks} \label{sec:irk_ork}

First of all, we are interested in upper bounds on the entropy of coding functions.

\begin{definition}
The {\em inner rank} and {\em outer rank} of a subset $X$ of vertices are respectively given by
\begin{align*}
    \irk(X) &:= \min\{|b|: \cl(X) = \cl(b)\}\\
    \ork(X) &:= \min\{|b| : X \subseteq \cl(b)\} = \min\{|b|: \cl(X) \subseteq \cl(b)\}.
\end{align*}
\end{definition}

Although the notations should reflect which closure operator is used in order to be rigorous, we shall usually omit this dependence for the sake of clarity. Instead, if the closure operator is ``decorated'' by subscripts or superscripts, then the corresponding parameters will be decorated in the same fashion.

A set $i$ with $|i| = \irk(X)$ and $\cl(i) = \cl(X)$ is called an {\em inner basis} of $X$; similarly a set $o$ with $|o| = \ork(X)$ and $\cl(X) \subseteq \cl(o)$ is called an {\em outer basis} of $X$.

The following properties are an easy exercise.

\begin{proposition} \label{prop:irk_ork}
For any $X,Y \subseteq V$,
\begin{enumerate}
    \item \label{it:rk(cl)} $\ork(\cl(X)) = \ork(X)$ and $\irk(\cl(X)) = \irk(X)$;

    \item \label{it:rk_scaled} $\ork(X) \le \irk(X) \le |X|$;

    \item \label{it:rk(cup)} $\ork(X \cup Y) \le \ork(X) + \ork(Y)$ and $\irk(X \cup Y) \le \irk(X) + \irk(Y)$;

    \item \label{it:rk(V)} $\ork(\emptyset) = \irk(\emptyset) = 0$ and $\ork(V) = \irk(V) = r$;

    \item \label{it:rk_increasing} if $X \subseteq Y$, then $\ork(X) \le \ork(Y)$.
\end{enumerate}
\end{proposition}

The closure of the empty set is the only closed set of (inner and outer) rank $0$, while $V$ is not necessarily the unique closed set of (inner or outer) rank $r$.

Note that the inner rank is not monotonic, as seen in the example in Figure \ref{fig:rk'}. We have $\cl_D(4) = V$ and hence $\irk_D(V) = 1$, while $\irk_D(123) = 2$ for $\cl_D(12) = 123$ while $\cl_D(v) = v$ for any $v \in 123$.

\begin{figure}
\begin{center}
\begin{tikzpicture}
    \tikzstyle{every node}=[draw,shape=circle];

	\node (1) at (3,2) {1};
	\node (2) at (3,0) {2};
	\node (3) at (4,1) {3};
	\node (4) at (2,1) {4};
	\node (5) at (0,1) {5};

    \draw[-latex] (1) -- (3);
    \draw[-latex] (2) -- (3);
    \draw[-latex] (3) -- (4);
    \draw[latex-latex] (4) -- (5);
    \draw[-latex] (4) -- (1);
    \draw[-latex] (4) -- (2);
\end{tikzpicture}
\end{center}
\caption{Example where the inner rank is not monotonic.} \label{fig:rk'}
\end{figure}
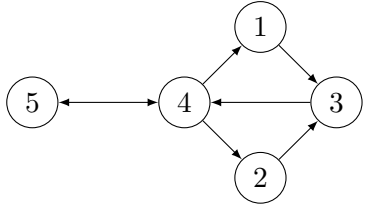

If $\cl_1(X) \subseteq \cl_2(X)$ for some $X$, then $\ork_1(X) \ge \ork_2(X)$. Indeed, any outer basis of $X$ with respect to $\cl_1$ is also an outer basis of $X$ with respect to $\cl_2$. In particular, if $\cl_1 \le \cl_2$, then $\ork_1(X) \ge \ork_2(X)$ for all $X$.

\begin{lemma} \label{lem:G<ork}
Let $G: 2^V \to \mathbb{R}$ satisfying $0 \le G(X) \le |X|$ and $G(\cl(X)) = G(X)$ for all $X \subseteq V$.  Then $G(X) \le \irk(X)$ for all $X$. Also, if $X \subseteq Y$ implies $G(X) \le G(Y)$, then $G(X) \le \ork(X)$ for all $X$.
\end{lemma}

\begin{proof}
First, if $i$ is an inner basis of $X$, then $G(X) = G(\cl(i)) = G(i) \le |i| = \irk(X)$. Second, if $o$ is an outer basis of $X$, $G(X) \le G(\cl(o)) = G(o) \le |o| = \ork(X)$.
\end{proof}

This Lemma proves that we get subadditivity for free. Since the entropy satisfies all the conditions of Lemma \ref{lem:G<ork}, we obtain an upper bound on the entropy. 

\begin{corollary} \label{cor:H<ork}
For any coding function $f$ and any $X \subseteq V$, $H_f(X) \le \ork(X)$.
\end{corollary}

\subsection{Flats and span} \label{sec:flats}

Before we move on to lower bounds on the entropy, we define two fundamental concepts.

\begin{definition} \label{def:flat}
A {\em flat} is a subset $F$ of vertices for which there is no $X \supset F$ with $\ork(X) = \ork(F)$.
\end{definition}

\begin{proposition} \label{prop:flat}
Flats satisfy the following properties.
\begin{enumerate}
	\item \label{it:flat_V} $\cl(\emptyset)$ is the only flat with rank $0$, and $V$ is the only flat with rank $r$.

    \item \label{it:flat_closed} Any flat $F$ is a closed set;

    \item \label{it:rk(F)} $\ork(F) = \irk(F)$;

    \item \label{it:flat_rank} for any $X$, there exists a flat $F \supseteq X$ with $\ork(F) = \ork(X)$.
\end{enumerate}
\end{proposition}

\begin{proof}
\ref{it:flat_V} is trivial.

\ref{it:flat_closed}. Since $\cl(F)$ contains $F$ while having the same rank as $F$, it cannot properly contain $F$.

\ref{it:rk(F)}. Let $o$ be an outer basis of $F$. Since $F \subseteq \cl(o)$ while $\ork(F) = \ork(\cl(o))$, we obtain $F = \cl(o)$ and $o$ is an inner basis of $F$.

\ref{it:flat_rank}. For any $X$, let $C$ be a set with rank $\ork(X)$ and containing $X$ of largest cardinality, then there exists no $G$ such that $C \subset G$ and $\ork(G) = \ork(X) = \ork(C)$.
\end{proof}

It is worth noting that there are closed sets which are not flats. For example, consider the following closure operator on $V = \{1,\ldots,n\}$, where $\cl(X) = \{1,\ldots,\max(X)\}$. Then it has rank $1$ and hence only two flats (the empty set and $V$), while it has $n+1$ closed sets (the empty set and $\cl(i)$ for all $i$). We shall clarify the relationship between closed sets and flats below.

\begin{definition} \label{def:span}
For any $X \subseteq V$, the union of all flats containing $X$ with outer rank equal to that of $X$ is referred to as the {\em span} of $X$, i.e.
$$
	\spancl(X) := \bigcup \{F: F \,\mbox{flat}, X \subseteq F, \ork(F) = \ork(X)\}.
$$
\end{definition}

\begin{proposition} \label{prop:span}
For any $X$,
\begin{enumerate}	
	\item  $\cl(X) \subseteq \spancl(X)$ with equality if and only if $\cl(X)$ is a flat;

	\item $\spancl(\cl(X)) = \spancl(X)$;
	
	\item $\spancl(X) := \{v \in V: \ork(X \cup v) = \ork(X)\}$.
\end{enumerate}
\end{proposition}

\begin{proof}
The first two properties follow directly from the definition. Suppose $v \in F$, a flat containing $X$ with $\ork(F) = \ork(X)$, then $\ork(X \cup v) \le \ork(F) = \ork(X)$. Conversely, if $\ork(X \cup w) = \ork(X)$, then $X \cup w$ is contained in a flat with the same outer rank as $X$, and hence in $\spancl(X)$.
\end{proof}

Flats and spans provide two alternate axioms for matroids.

\begin{theorem} \label{th:matroid}
The following are equivalent:
\begin{enumerate}
	\item $\cl$ is a matroid;

	\item all closed sets are flats, i.e. $\cl(X) = \spancl(X)$ for all $X \subseteq V$;

	\item all closed sets are spans, i.e. for all $X \subseteq V$, there exists $Y \subseteq V$ such that $\cl(X) = \spancl(Y)$.
\end{enumerate}
\end{theorem}

\begin{proof}
The first property clearly implies the third one. Let us now prove that the second property implies the first one. Let $X \subseteq V$, $v \in V$ and $u \in \cl(X \cup v) \backslash \cl(X)$, then $\ork(X \cup u) = \ork(X)+1 = \ork(X \cup v)$, and hence $\cl(X \cup u) = \cl(X \cup v)$. Thus, $\cl$ satisfies the Steinitz-Mac Lane exchange axiom.

We now prove that the third property implies the second. Suppose all closed sets are spans, then we shall prove that all closed sets of outer rank $k$ are flats, by induction on $0 \le k \le r$. This is clear for $k=0$, hence suppose it holds for up to $k-1$. Consider a minimal closed set $c$ of outer rank $k$, i.e. $\ork(c) = k$ and $\ork(c') = k-1$ for any closed set $c' \subset c$. By hypothesis, we have $c = \spancl(Y)$ for some $Y \subseteq c$; we now prove that $c = \cl(Y)$. Suppose that $\cl(Y) \subset c$, then $\cl(Y) = c'$, a closet set of outer rank at most $k-1$. Then $c'$ is a flat, i.e. $c' = \spancl(c') = \spancl(\cl(Y)) = \spancl(Y) = c$, contradiction. Thus, $c = \spancl(c)$ and $c$ is a flat.
\end{proof}

There are solvable closure operators which are not matroids, e.g. the undirected graph $\bar{C}_4$ displayed in Figure \ref{fig:C4}. It is solvable because it has rank $2$ and contains $K_2 \cup K_2$. More explicitly, the following is a solution for it over any alphabet:
\begin{align*}
	P_i(y) &:= \{x \in A^2: x_i = y\}, \quad i \in \{1,2\}\\
         f_1 = f_3 &= \{P_1(y) : y \in A\}\\
         f_2 = f_4 &= \{P_2(y) : y \in A\}\\
\end{align*}
 In that case, note that the outer rank is submodular, and hence $\spancl_{\bar{C}_4} = U_{2,4}$ is a matroid; however, $\cl_{\bar{C}_4}$ is not a matroid itself.

We would like to explain the significance of flats in matroids for random network coding. A model for noncoherent random network coding based on matroids is proposed in \cite{GG11}, which generalises routing (a special case for the uniform matroid), linear network coding (the projective geometry) and affine network coding (the affine geometry). In order to combine the messages they receive, the intermediate nodes select a random element from the closure of the received messages. The model is based on matroids because all closed sets are flats, hence a new message is either in the closure of all the previously received messages (and is not informative), or it increases the outer rank (and is fully informative).

\begin{figure}
\begin{center}
\begin{tikzpicture}
    \tikzstyle{every node}=[draw,shape=circle];

	\node (1) at (0,2) {1};
	\node (2) at (2,2) {2};
	\node (3) at (2,0) {3};
	\node (4) at (0,0) {4};

    \draw[latex-latex] (1) -- (2);
    \draw[latex-latex] (2) -- (3);
    \draw[latex-latex] (3) -- (4);
    \draw[latex-latex] (4) -- (1);
\end{tikzpicture}
\end{center}
\caption{The graph $\bar{C}_4$ whose closure operator is solvable but not a matroid.} \label{fig:C4}
\end{figure}
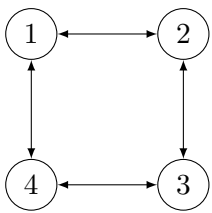

\subsection{Upper and lower ranks} \label{sec:lr_ur}

We are now interested in lower bounds on the entropy of coding functions. Since any closure operator has a trivial coding function with entropy zero (where the universal partition is placed on every vertex), the entropy of any coding function cannot be bounded below. Therefore, most of our bounds will apply to solutions only.

\begin{definition} \label{def:lr_ur}
The {\em lower rank} and {\em upper rank} of $X$ are respectively defined as
\begin{align*}
	\lrk(X) &:= \min\{|Y| : \cl(Y \cup (V \backslash X)) = V\},\\
	\urk(X) &:= r - \lrk(V \backslash X).
\end{align*}
\end{definition}

A few elementary properties of the lower and upper ranks are listed below. Again, if $\cl_1 \le \cl_2$, then $\lrk_1(X) \ge \lrk_2(X)$ and $\urk_1(X) \ge \urk_2(X)$ for all $X \subseteq V$.

\begin{lemma} \label{lem:lr_ur}
The following hold:
\begin{enumerate}
	\item $\lrk(V) = \urk(V) = r$ and $\lrk(\emptyset) = \urk(\emptyset) = 0$.
	
	\item For any $X \subseteq V$, $\lrk(X) = 0$ if and only if $\cl(V \backslash X) = V$. Hence $\urk(X) = r$ if and only if $\cl(X) = V$.

	\item \label{it:urk_alternate} For any $X \subseteq V$,
	\begin{align*}
		\urk(X) &= r - \min\{\ork(Y): \cl(X \cup Y) = V\}\\ &= r - \min\{\ork(F): F \mbox{ flat and } \cl(X \cup F) = V\}.
	\end{align*}

	\item \label{it:urk(X)<urk(Z)} If $X \subseteq Z$, then $\urk(X) \le \urk(Z)$ and $\lrk(X) \le \lrk(Z)$.

	\item \label{it:urk_cl} $\urk(X) = \urk(\cl(X))$.
	
	\item \label{it:lrk<urk} $\lrk(X) \le \urk(X) \le \ork(X)$.
\end{enumerate}
\end{lemma}

\begin{proof}
The first three properties are easily proved. Property \ref{it:urk(X)<urk(Z)} for the upper rank follows from Property \ref{it:urk_alternate}; the result for the lower rank follows from $\lrk(X) = r- \urk(V \backslash X)$. For Property \ref{it:urk_cl}, Property \ref{it:urk(X)<urk(Z)} yields $\urk(X) \le \urk(\cl(X))$ while $\cl(X \cup Y) = \cl(\cl(X) \cup Y)$ yields the reverse inequality. We now prove Property \ref{it:lrk<urk}. The inequality $\urk(X) \le \ork(X)$ follows from the subadditivity of the outer rank. To prove that $\lrk(X) \le \urk(X)$, let $b$ be a basis for $\cl$. Then
$$
	V = \cl(b) = \cl\left\{(b \cap X) \cup (b \cap (V \backslash X))\right\} \subseteq \cl\left\{(b \cap X) \cup (V \backslash X)\right\},
$$
and hence $\cl\left\{(b \cap X) \cup (V \backslash X)\right\} = V$, thus $|b \cap X| \ge \lrk(X)$. Similarly, $|b \cap (V \backslash X)| \ge \lrk(V \backslash X)$, and hence
$r = |b| \ge \lrk(X) + \lrk(V \backslash X).$
\end{proof}

We remark that for any solution $f$, we have $r = H_f(V) \le \ork(X) + \ork(Y)$ for any $X,Y$ such that $\cl(X \cup Y) = V$. Therefore, we obtain
$$
	H_f(X) \ge \urk(X)
$$
for all $X \subseteq V$.

\begin{corollary} \label{cor:bounds_H}
For any solution $f$ of $\cl$ and any $X \subseteq V$,
$$
	r - H_f(V \backslash \cl(X)) \le \lrk(\cl(X)) \le \urk(X) \le H_f(X) \le \ork(X).
$$
\end{corollary}

Note that a trivial lower bound on $H_f(X)$ (where $f$ is a solution) is given by $r - H_f(V \backslash X)$. Therefore, the intermediate bounds on $H_f(X)$ in Corollary \ref{cor:bounds_H} refine this trivial bound.

Some of the results above can be generalised for any coding function $f$: denoting
\begin{align*}
	\lrk_f(X) &= \min\{H_f(Y) : \cl(Y \cup (V \backslash X)) = V\},\\
	\urk_f(X) &= H_f(V) -\lrk_f(V \backslash X),
\end{align*}
we obtain
$$
	H_f(V) - H_f(V \backslash \cl(X)) \le \lrk_f(\cl(X)) \le \urk_f(X) \le H_f(X) \le \ork(X).
$$

We finish this subsection by remarking that Theorem \ref{th:matroid} has an analogue for the upper rank. Namely, define an upper flat as a set $F$ such that $F \subset X$ implies $\urk(X) > \urk(F)$; define also the upper span of $X$ as
$$
	\uspancl(X) := \bigcup \{F: F \,\mbox{upper flat}, X \subseteq F, \urk(F) = \urk(X)\} = \{v \in V: \urk(X \cup v) = \urk(X)\}.
$$

\begin{theorem}
The following are equivalent:
\begin{enumerate}
	\item $\cl$ is a matroid;

	\item all closed sets are upper flats, i.e. $\cl(X) = \uspancl(X)$ for all $X \subseteq V$;

	\item all closed sets are upper spans, i.e. for all $X \subseteq V$, there exists $Y \subseteq V$ such that $\cl(X) = \uspancl(Y)$.
\end{enumerate}
\end{theorem}

\subsection{Inner and outer complemented sets} \label{sec:complemented}

We are now interested in a case where the bounds on the entropy are tight.

\begin{definition} \label{def:complemented}
We say a set $X$ is {\em outer complemented} if $\ork(X) = \urk(X)$. Moreover, we say it is {\em inner complemented} if $\irk(X) = \urk(X)$.
\end{definition}

Therefore, if $X$ is outer complemented, then $H_f(X) = \ork(X) = \urk(X)$ for any solution $f$.

Remark that $X$ is outer (inner) complemented if and only if $\cl(X)$ is outer (inner) complemented.

\begin{proposition} \label{prop:complemented}
The following are equivalent:
\begin{enumerate}
    \item $X$ is outer complemented;

    \item \label{it:oc_Y} there exists $Z$ such that $\ork(X) + \ork(Z) = r$, $\cl(X \cup Z) = V$ and $X \cap Z = \emptyset$;

%    \item there exists a flat $F$ such that $\ork(X) + \irk(F) = r$ and $\cl(X \cup F) = V$;

    \item any outer basis of $X$ is contained in a basis of $V$.
\end{enumerate}

Similar results hold for inner complemented sets. The following are equivalent:
\begin{enumerate}
    \item $X$ is inner complemented;

    \item $X$ is outer complemented and $\irk(X) = \ork(X)$;

    \item any inner basis of $X$ is contained in a basis of $V$.
\end{enumerate}
\end{proposition}

\begin{proof}
The equivalence of the first two properties is easily shown. If $X$ is outer complemented, let $o$ be an outer basis of $X$ and let $Z$ satisfy $\cl(X \cup Z) = V$ and $|Z| = r - \ork(X)$. Then $o \cup Z$ is a basis of $V$. Conversely, if any outer basis can be extended to a basis, then any such extension is a valid $Z$ for Property \ref{it:oc_Y}. The properties for an inner complemented set are easy to prove.
\end{proof}

We saw earlier that $\cl(X) \subseteq \cl_f(X)$ for any coding function $f$ and any $X$. This can be refined when $f$ is a solution and $X$ is outer complemented.

\begin{lemma} \label{prop:complemented_span}
If $f$ is a solution of $\cl$ then $\cl(\spancl(X)) \subseteq \cl_f(X)$ for any outer complemented $X$.
\end{lemma}

\begin{proof}
For any outer complemented $X$, we have $H_f(X) = \ork(X)$. Suppose $v \in \spancl(X)$, then $\ork(X) = \ork(X \cup v) \ge H_f(X \cup v) \ge H_f(X) = \ork(X)$ and hence $v \in \cl_f(X)$. Since $\cl_f(X)$ is a closed set of $\cl$, we easily obtain that $\cl(\spancl(X)) \subseteq \cl_f(X)$.
\end{proof}

\begin{corollary} \label{cor:complemented_span}
If there exists an outer complemented set $X$ such that its span has higher outer rank and is also outer complemented, then $\cl$ is not solvable over any alphabet.
\end{corollary}

By extension, we say that $\cl$ is outer complemented if all sets are outer complemented. We can characterise the solvable outer complemented closure operators.

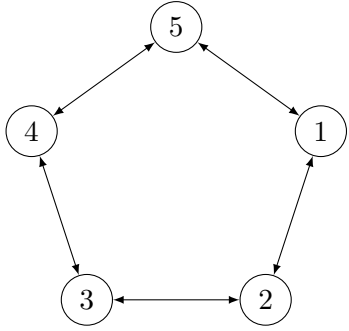
\begin{figure}
\begin{center}
\begin{tikzpicture}
    \tikzstyle{every node}=[draw,shape=circle];

\foreach \x in {1,...,5}{
	\node (\x) at (90-72*\x:2) {\x};
}

    \draw[latex-latex] (1) -- (2);
    \draw[latex-latex] (2) -- (3);
    \draw[latex-latex] (3) -- (4);
    \draw[latex-latex] (4) -- (5);
    \draw[latex-latex] (5) -- (1);
\end{tikzpicture}
\end{center}
\caption{The graph $\bar{C}_5$ whose closure operator is outer complemented and not solvable.} \label{fig:C5}
\end{figure}

\begin{theorem} \label{th:complemented_span_matroid}
Suppose that $\cl$ has rank $r$ and is outer complemented. Then $\cl$ is solvable if and only if $\spancl$ is a solvable matroid with rank $r$.
\end{theorem}

\begin{proof}
If all sets are outer complemented, then any solution $f$ of $\cl$ is also a coding function of $\spancl$ since $\spancl(X) = \{v \in V : H_f(X \cup v) = H_f(X)\}$. Since the outer rank is equal to the entropy $H_f$, it is submodular and hence $\spancl$ is a matroid whose rank function is given by the outer rank. Thus $\spancl$ has rank $r$ and $f$ is a solution for it. Conversely, if $\spancl$ is a solvable matroid with rank $r$, then we have $\cl \le \spancl$ and $\cl$ is solvable.
\end{proof}

For instance, for the undirected cycle $\bar{C}_5$ in Figure \ref{fig:C5}, $\cl_{\bar{C}_5}$ is outer complemented, though the outer rank is not submodular, hence $\spancl$ is not a matroid. As such, $\bar{C}_5$ is not solvable (its entropy is actually $2.5$ \cite{Rii07}).

We would like to emphasize that if all sets are outer complemented, then the outer rank must be submodular, i.e. the rank function of a matroid. However, this does not imply that $\cl$ should be a matroid itself. For instance, consider $\cl$ defined on $\{1,2,3\}$ as follows: $\cl(1) = 12$, $\cl(2) = 2$, $\cl(3) = 3$, $\cl(13) = \cl(23) = 123$. Then any set is inner complemented, $\cl$ is solvable (by letting $f_1 = f_2$ and $f_3$ such that $f_1 \vee f_3 = E_{A^2}$) but $\cl$ is not a matroid.

\subsection{Combining closure operators} \label{sec:combining}

In this subsection, $V_1$ and $V_2$ are disjoint sets of respective cardinalities $n_1$ and $n_2$; $\cl_1$ and $\cl_2$ are closure operators on $V_1$ with rank $r_1$ and on $V_2$ with rank $r_2$, respectively. We further let $V = V_1 \cup V_2$ and for any $X \subseteq V$, we shall denote $X_1 = X \cap V_1$ and $X_2 \cap V_2$. Different ways of combining closure operators have been proposed in \cite{Gad13}.

\begin{definition} \label{def:unions}
The disjoint, unidirectional, and bidirectional unions of $\cl_1$ and $\cl_2$ are respectively
\begin{align*}
	\cl_1 \cup \cl_2(X) &:= \cl_1(X_1) \cup \cl_2(X_2)\\
	\cl_1 \vec{\cup} \cl_2(X) &:= \begin{cases}
	V_1 \cup \cl_2(X_2) &\mbox{if } \cl_1(X_1) = V_1\\
	\cl_1(X_1) \cup X_2 &\mbox{otherwise}
	\end{cases}\\
	\cl_1 \bar{\cup} \cl_2(X) &:= \begin{cases}
	V_1 \cup \cl_2(X_2) & \mbox{if } X_1 = V_1\\
	V_2 \cup \cl_2(X_1) & \mbox{if } X_2 = V_2\\
	X & \mbox{otherwise}.
	\end{cases}
\end{align*}
\end{definition}

If $\cl$ is a closure operator on $V$ satisfying $\cl_1 \vec{\cup} \cl_2 \le \cl \le \cl_1 \cup \cl_2$, it has rank $r_1 + r_2$ and entropy $H(\cl_1) + H(\cl_2)$. We can then split the problems into two parts. In that case, we also have 
$$
	\cl_1(X_1) = \cl(X) \cap V_1 = \cl(X_1 \cup V_2) \cap V_1 = \cl(X_1) \cap V_1
$$ 
for all $X \subseteq V$, i.e. $V_2$ has no influence on $\cl(X)$ on $V_1$.

The rank of the bidirectional union is given by
$$
	r(\cl_1 \bar{\cup} \cl_2) = \min\{n_1 + r_2, n_2 + r_1\},
$$
while its entropy only satisfies the inequality
$$
	H(\cl_1 \bar{\cup} \cl_2) \le \min\{n_1 + H(\cl_2), n_2 + H(\cl_1)\}.
$$

We can determine how the four rank functions given above behave with regards to the three types of union.

\begin{proposition}
For the disjoint union, let $\cl_\cup := \cl_1 \cup \cl_2$, then
\begin{align*}
	r_\cup &= r_1 + r_2\\
	\ork_\cup(X) & = \ork_1(X_1) + \ork_2(X_2)\\
	\irk_\cup(X) & = \irk_1(X_1) + \irk_2(X_2)\\
	\urk_\cup(X) & = \urk_1(X_1) + \urk_2(X_2)\\
	\lrk_\cup(X) & = \lrk_1(X_1) + \lrk_2(X_2).
\end{align*}

For the unidirectional union, let $\cl_{\vec{\cup}} := \cl_1 \vec{\cup} \cl_2$, then
\begin{align*}
	r_{\vec{\cup}} &= r_1 + r_2\\
	\ork_{\vec{\cup}}(X) & = \min\{r_1 + \ork_2(X_2), \ork_1(X_1) + |X_2|\}\\
	\irk_{\vec{\cup}}(X) & = \min\{r_1 + \irk_2(X_2), \irk_1(X_1) + |X_2|\}\\
	\urk_{\vec{\cup}}(X) & = \urk_1(X_1) + \urk_2(X_2)\\
	\lrk_{\vec{\cup}}(X) & = \lrk_1(X_1) + \lrk_2(X_2).
\end{align*}

For the bidirectional union, let $\cl_{\bar{\cup}} := \cl_1 \bar{\cup} \cl_2$, then
\begin{align*}
	r_{\bar{\cup}} &= \min\{n_1 + r_2, n_2 + r_1\}\\
	\ork_{\bar{\cup}}(X) & = \min\{n_1 + \ork_2(X_2), n_2 + \ork_1(X_1), |X|\}\\
	\irk_{\bar{\cup}}(X) & = 
	\begin{cases}
		\min\{n_1 + r_2, n_2 + r_1\} & \mbox{if } \cl(X) = V\\
		n_1 + \irk_2(X_2) & \mbox{if } X_1 = V_1, \cl_2(X_2) \ne V_2\\
		n_2 + \irk_1(X_1) & \mbox{if } X_2 = V_2, \cl_1(X_1) \ne V_1\\
		|X| & \mbox{otherwise}
	\end{cases}\\
	\urk_{\bar{\cup}}(X) & = r_{\bar{\cup}} - \min\{n_1 - |X_1| + r_2 - \urk_2(X_2), n_2 - |X_2| + r_1 - \urk_1(X_1)\}\\
	\lrk_{\bar{\cup}}(X) & = \min\{|X_1| + \lrk_2(X_2), |X_2| + \lrk_1(X_1)\}.
\end{align*}
\end{proposition}

\begin{proof}
The results for the disjoint union easily follow from the definitions. We then turn to the unidirectional union. Again, we remark that $\cl_{\vec{\cup}}(X) = V$ if and only if $\cl_1(X_1) = V_1$ and $\cl_2(X_2) = V_2$; this gives the rank, the upper rank of $X$ and then its lower rank. For the upper rank, we have $X \subseteq \cl_{\vec{\cup}}(o)$ if and only if either $\cl_1(o_1) = V_1, X_2 \subseteq \cl_2(o_2)$ or $X_1 \subseteq \cl_1(o_1), X_2 \subseteq o_2$. The proof for the inner rank is similar.

For the bidirectional union, we have $X \subseteq \cl_{\bar{\cup}}(o)$ if and only if $o_1 = V_1, X_2 \subseteq \cl_2(o_2)$ or $o_2 = V_2, X_1 \subseteq \cl_1(o_1)$ or $X \subseteq o$; this yields the outer rank. The inner rank is obtained by considering each case separately. If $X=V$, then $\cl_{\bar{\cup}}(i) = \cl_{\bar{\cup}}(X) = V$ if and only if $i_1 = V_1, \cl_2(i_2) = V_2$ or $i_2 = V_2, \cl_1(i_1) = V_1$. If $X_1 = V_1$ but $\cl_2(X_2) \ne V_2$, then $\cl_{\bar{\cup}}(i) = \cl_{\bar{\cup}}(X)$ if and only if $i_1 = V_1, \cl_2(i_2) = \cl_2(X_2)$. The third case comes from symmetry. Finally, if $X_1 \ne V_1$ and $X_2 \ne V_2$, then $\cl_{\bar{\cup}}(i) = \cl_{\bar{\cup}}(X) = X$ if and only if $i = X$. For the upper rank, we remark that $\cl_{\bar{\cup}}(X \cup Y) = V$ if and only if either $X_1 \cup Y_1 = V_1, \cl_2(X_2 \cup Y_2) = V_2$ or $X_2 \cup Y_2 = V_2, \cl_1(X_1 \cup Y_1) = V_1$. The lower rank follows from the upper rank.
\end{proof}

\section{Shannon entropy} \label{sec:Shannon}

Since finding the entropy of a digraph is difficult in general, \cite{Rii06, Rii07} develops the idea of Shannon entropy of a graph. The main idea is to maximise over any function which satisfies some of the properties of an entropic function, notably submodularity. This idea can be adapted to general closure operators. For any closure operator $\cl$ on $V$, a Shannon function for $\cl$ can be viewed as a $\cl$-compatible polymatroid.

\begin{definition}
For any closure operator $\cl$ on $V$, a {\em Shannon function} for $\cl$ is a function $r : 2^V \to \mathbb{R}$ such that
\begin{enumerate}
    \item \label{it:SE_scaled} if $X \subseteq V$, then $$0 \leq r(X) \leq |X|,$$

    \item \label{it:SE_increasing} $r$ is increasing, i.e. if $X \subseteq Y \subseteq V$, then $$r(X) \leq r(Y),$$

    \item \label{it:SE_submodular} $r$ is submodular, i.e. if $X,Y \subseteq V$, then
    \begin{equation} \label{eq:submodular} \nonumber
        r(X) + r(Y) \geq r(X \cup Y) + r(X \cap Y),
    \end{equation}

    \item \label{it:SE_cl} for all $X \subseteq V$, $$r(X) = r(\cl(X)).$$
\end{enumerate}
The maximum value of $r(V)$ over all Shannon functions for $\cl$ is called the {\em Shannon entropy} of $\cl$ and is denoted by $\SE(\cl)$.
\end{definition}

Any Shannon function also satisfies the conditions of Lemma \ref{lem:G<ork}, hence $\SE(\cl) \le \ork(V) = r$. Moreover, it is clear that if $\cl_1 \le \cl_2$, then $\SE(\cl_1) \ge \SE(\cl_2)$.

\subsection{Shannon entropy and combining closure operators}

\begin{lemma} \label{lem:S_weak}
If $\cl_1 \vec{\cup} \cl_2 \le \cl \le \cl_1 \cup \cl_2$, then for any Shannon function $r$ for $\cl$, the function
$$
	r'(X) := r(X \cap V_1) + r(X \cup V_1) - r(V_1)
$$
is a Shannon function for $\cl$ such that $r'(X) = r'(X_1) + r'(X_2)$ and $r'(V) = r(V)$.
\end{lemma}

\begin{proof}
Only the closure property is nontrivial to verify. Since $\cl(X) \cap V_1 = \cl(X_1) \cap V_1$, we obtain
$$
	r'(\cl(X)) = r(\cl(X) \cap V_1) + r(\cl(X) \cup V_1) - r(V_1) \le r(\cl(X_1)) + r(\cl(X \cup V_1)) - r(V_1) = r'(X).
$$
\end{proof}

\begin{proposition}
If $\cl_1 \vec{\cup} \cl_2 \le \cl \le \cl_1 \cup \cl_2$, then
$$
    \SE(\cl) = \SE(\cl_1) + \SE(\cl_2).
$$
\end{proposition}

\begin{proof} First of all, it is clear that $\SE(\cl) \ge \SE(\cl_1) + \SE(\cl_2)$. Indeed, let $r_i$ be a Shannon function for $\cl_i$, then $r$ defined by $r(X) := r_1(X_1) + r_2(X_2)$ is a Shannon function for $\cl_1 \cup \cl_2$.

We now show the reverse inequality. By Lemma \ref{lem:S_weak}, there exists a Shannon function $r$ for $\cl$ with $r(X) = r(X_1) + r(X_2)$ and $r(V) = \SE(\cl)$. It is easily seen that the restriction $r_2(X)$ of $r(X)$ to $V_2$ is a Shannon function for $\cl_2$, hence $r_2(V_2) = r(V_2) \le \SE(\cl_2)$. Also, define the function $r_1: 2^{V_1} \to \mathbb{R}$ as
$$
    r_1(X) := r(X \cup V_2) - r(V_2).
$$
We check that $r_1$ is indeed a Shannon function for $\cl_1$. The first two properties are straightforward, while submodularity comes from
\begin{eqnarray*}
    r_1(X) + r_1(Y) &=& r(X \cup V_2) + r(Y \cup V_2) - 2r(V_2)\\
    &\ge& r((X \cup Y) \cup V_2) + r((X \cap Y) \cup V_2) - 2r(V_2)\\
    &=& r_1(X \cup Y) + r_1(X \cap Y)
\end{eqnarray*}
and the closure property comes from the fact that $\cl_1(X) = \cl(X \cup V_2) \backslash V_2$.
$$
    r_1(\cl_1(X)) = r(\cl_1(X) \cup V_2) - r(V_2) = r(\cl(X \cup V_2)) - r(V_2) = r(X \cup V_2) - r(V_2) =  r_1(X).
$$
Thus,
$$
    r(V) = r_1(V_1) + r_2(V_2) \le \SE(\cl_1) + \SE(\cl_2).
$$
\end{proof}

The Shannon entropy of the bidirectional union satisfies a similar inequality to the one for the corresponding entropy.

\begin{proposition}
For any $\cl_1$ and $\cl_2$, we have
$$
	\SE(\cl_1 \bar{\cup} \cl_2) \le \min\{\SE(\cl_1) + n_2, \SE(\cl_2) + n_1\}.
$$
\end{proposition}

\begin{proof}
We say a function $r: 2^V \to 2^V$ is a $V_1$-function if it satisfies all the properties of a Shannon function but only for all $X,Y$ containing $V_1$. The maximum value of $r(V)$ over any $V_1$-function, denoted as $S$, is greater than or equal to $\SE(\cl_1 \bar{\cup} \cl_2)$. Let $r$ be a $V_1$-function and consider
$$
	r_2(X) := r(X \cup V_1) - r(V_1)
$$
for all $X \subseteq V_2$. We then prove that $r_2$ is a Shannon function for $\cl_2$. Only Property \ref{it:SE_cl} is nontrivial to check: we have
$$
	r_2(\cl_2(X)) = r(\cl_2(X) \cup V_1) - r(V_1) = r(\cl_1 \bar{\cup} \cl_2(X \cup V_1)) - r(V_1) = r(X \cup V_1) - r(V_1) = r_2(X).
$$
If $r$ achieves $r(V) = S$, we obtain $S \le r_2(V_2) + r(V_1) \le \SE(\cl_2) + n_1$. Thus, $\SE(\cl_1 \bar{\cup} \cl_2) \le S \le \SE(\cl_2) + n_1$. Symmetry finished the proof.
\end{proof}

\subsection{Density of closure entropies}

We remark that any closure operator of rank at least one has entropy at least one (assign the universal partition to every vertex in $\cl(\emptyset)$ and the same partition $g$ of $A^r$ into $|A|$ parts to any other vertices). Moreover, any $D$-closure for a digraph $D$ with rank (i.e. minimum feedback vertex set size) of $2$ has entropy $2$; in fact, such closure operators are solvable over any sufficient large alphabet \cite{Gad13}. This shows that multiple unicast instances with two source-destination pairs are solvable over all sufficiently large alphabets. This proof technique cannot be generalised for general digraphs, for $\bar{C}_5$ has rank $3$ but entropy only $2.5$. Another direction could then be to consider other families of closure operators and find ``gaps'' in the entropy distribution; in particular, we may ask whether all closure operators of rank $2$ are solvable. Theorem \ref{th:entropy_dense} gives an emphatic negative answer to the last question: the set of all possible closure entropies is always dense above $1$. 

\begin{theorem} \label{th:entropy_dense}
For any $r \ge 2$ and any rational number $H$ in $(1,r]$, there exists a closure operator of rank $r$ with entropy equal to $H$.
\end{theorem}

The proof is constructive, i.e. for any $H$, we give a closure operator with entropy equal to $H$ and the corresponding coding functions with entropy $H$.

First of all, we introduce some notation regarding rooted trees. A rooted tree is a tree with a specific vertex, called the root, denoted as $R$. The vertices at distance $k$ from the root form level $k$ of the tree (hence the root is the only vertex on level $0$), this is denoted $l_k$. For any vertex $v$ in level $k$, its parent is the only vertex adjacent to $v$ on level $k-1$; we denote it as $p(v)$. Moreover, we denote its ancestry as $a(v) := \{v, p(v), p^2(v), \ldots, R\}$ (remark that we include $v$ in its ancestry). Conversely, a child of $v$ is any vertex on level $k+1$ adjacent to $v$, any vertex without any children is a leaf of the tree. We denote the set of children of $v$ as $c(v)$. We extend the definitions above to any set of vertices $X$, e.g. $p(X) = \bigcup_{v \in X} p(v)$. The following properties easily follow.
\begin{enumerate}
	\item \label{it:aaX} If $u \in a(v)$, then $a(u) \subseteq a(v)$. Therefore, $a(a(X)) = a(X)$ for all $X$.

	\item \label{it:aX_in_aY} If $X \subseteq Y$, then $a(X) \subseteq a(Y)$.
	
	\item \label{it:cv_in_aX} $c(v) \subseteq a(X)$ only if $|X| \ge |c(v)|$.
\end{enumerate}

\begin{definition}
An $(L,C)$-tree, with $0 \le L \le C$ is a rooted tree with root $R$ with $L+1$ levels, such that any vertex of level $k$ has $C-k$ children for $0 \le k \le L-1$. If $L=0$, this tree reduces to a single vertex.
\end{definition}

We then have:
\begin{enumerate}
	\setcounter{enumi}{3}

	\item \label{it:children} Each vertex of level $L-1$ has $C-L+1$ children (which are leaves).

	\item \label{it:lk} For all $k$, $|l_k| = \frac{C!}{(C - k)!}$.
\end{enumerate}
 
We can express the rational number $H$ as 
$$
	H = \frac{1}{D} \sum_{t=1}^r N_t,
$$ 
where $D \ge \frac{r-1}{H-1}$, $0 < N_t < D$ for all $1 \le t \le r-1$, and $N_r = D$. We then introduce the following closure operator. Consider $r$ disjoint trees $T_1,\ldots,T_r$, where $T_t$ is an $(L_t := D-N_t , C_t := DH - N_t )$-tree with root $R_t$ for all $1 \le t \le r$. Then $V$ is the set of all vertices of the $r$ trees and
$$
	\cl(X) := \begin{cases}
	V & \mbox{if } \exists v : c(v) \subseteq a(X) \, \mbox{or}\, X \cap T_t \ne \emptyset \,\forall 1 \le t \le r\\
	a(X) & \mbox{otherwise}.
	\end{cases}
$$

\begin{lemma} \label{lem:closure}
The operator $\cl$ is a closure operator of rank $r$.
\end{lemma}

\begin{proof}
We first prove that this is indeed a closure operator. First of all, it is trivial to check that $X \subseteq \cl(X)$. Secondly, if $X \subseteq Y$, then we need to check that $\cl(X) \subseteq \cl(Y)$. If $\cl(Y) = V$, this is trivial, otherwise $\cl(Y) = a(Y) \ne V$ and hence $\cl(X) = a(X) \subseteq \cl(Y)$ by Proposition \ref{it:aX_in_aY}. Thirdly, we need to prove that $\cl$ is idempotent. Again, this is trivial if $\cl(X) = V$, hence let $\cl(X) = a(X) \ne V$. By definition, there exists $t$ such that $X \cap T_t = \emptyset$, and hence $a(X) \cap T_t = \emptyset$. Also, for any non-leaf $v$, there exists a child $u$ of $v$ which does not belong to $a(X)$; then $u$ does not belong to $a(a(X)) = a(X)$ either. Therefore, we have $\cl(a(X)) = a(a(X)) = a(X)$ by Property \ref{it:aaX} and hence $\cl(\cl(X)) = \cl(a(X)) = a(X) = \cl(X)$.

We now prove that it has rank $r$. Since the set of roots has cardinality $r$ and intersects all trees, the rank is at most $r$. Conversely, suppose $\cl(X) = V$. Firstly, if there exists $v$ such that $c(V) \subseteq a(X)$, then $|X| \ge |c(v)|$ by Property \ref{it:cv_in_aX}; thus $|X| \ge C_t - L_t + 1 = D(H-1) + 1 \ge r$. Secondly, if $X$ intersects all trees, then $|X| \ge r$. Thirdly, if $\cl(X) = a(X)$, then $a(X) = V$ and $X$ intersects all trees; thus $|X| \ge r$. 
\end{proof}

\begin{lemma} \label{lem:entropy<H}
The entropy of $\cl$ is at most $H$.
\end{lemma}

\begin{proof}
The proof uses the submodular inequality recursively on all levels of a tree, and then successively for all trees. More precisely, we shall use the following application of the submodular inequality: if $r: 2^V \to 2^V$ is submodular and $X_1,\ldots,X_k$ are subsets of $V$ such that $X_i \cap X_j = X$ for all $i \ne j$ and $\bigcup_i X_i = Y$, we have
$$
	r(Y) \le \sum_{i=1}^k r(X_i) - (k-1)r(X).
$$

Fix a coding function $f$ for $\cl$. For any non-leaf $v$, the submodular inequality gives (with the sets $X_1,\ldots,X_k$ corresponding to $\{u,v\} : u \in c(v)$ and hence $X = v$, $Y = v \cup c(v)$)
$$
	H_f(V) = H_f(c(v)) \le \sum_{u \in c(v)} H_f(u) - (|c(v)| - 1) H_f(v).
$$
We first add up by level; for level $k$ ($0 \le k \le L_t$) of tree $T_t$ we denote $H_k := \sum_{v \in l_k} H_f(v)$ and we obtain
\begin{align*}
	\frac{C_t!}{(C_t-k)!} H_f(V) &\le H_{k+1} - (C_t-k-1) H_k\\
	\frac{C!}{(C_t-k)!} (C_t-k-2)! H_f(V) &\le (C_t-k-2)! H_{k+1} - (C_t-k-1)! H_k
\end{align*}
Let us now add up for all levels:
\begin{align*}
	\left[ \sum_{k=0}^{L_t-1} \frac{C_t!}{(C_t-k)(C_t-k-1)}\right] H_f(V) 
	&\le (C_t-L_t-1)! H_{L_t} - (C_t-1)! H_0\\ 
	&\le \frac{C_t!}{C_t-L_t} - (C_t-1)! H_0,
\end{align*}
where we used $H_{L_t} \le |l_{L_t}| = \frac{C_t!}{(C_t - L_t)!}$. Simplifying, we obtain
$$
	L_t H(V) \le C_t - (C_t-L_t)H_0 = C_t - D(H-1) H_f(R_t)
$$
since by definition $H_0 = H_f(R_t)$.We now add up for all trees $T_1$ up to $T_{r-1}$, we obtain
\begin{equation} \label{eq:H1}
	\left[ D(r-1) - D(H-1)\right] H(V) \le DH(r-1) - D(H-1) - D(H-1) \sum_{t=1}^{r-1} H_f(R_t),
\end{equation}
where we used the following relations:
\begin{align*}
	\sum_{b=1}^{r-1} N_t &= D(H-1)\\
	\sum_{b=1}^{r-1} L_t &= D(r-1) - D(H-1)\\
	\sum_{b=1}^{r-1} C_t &= DH(r-1) - D(H-1).
\end{align*}
Moreover, the set of all roots is a basis for the closure operator, hence 
\begin{equation}\label{eq:H2}
	H_f(V) = H_f(\{R_1,\ldots,R_r\}) \le 1 + \sum_{t=1}^{r-1} H_f(R_t).
\end{equation}
Multiplying \eqref{eq:H2} by $D(H-1)$ and adding it with \eqref{eq:H1} eventually yields
$$
	H_f(V) \le \frac{DH(r-1)}{D(r-1)} = H.
$$
\end{proof}

We now construct the coding function with entropy $H$. Consider $A = B^D$, where $B$ is any finite set of cardinality at least $2$. Any $x \in A^r$ can be expressed as $x = (x_1,\ldots,x_{rD}) \in B^{rD}$. For any subset $S \subseteq \{1,\ldots,rD\}$, say $S = \{i_1,\ldots,i_{|S|}\}$ once sorted in increasing order, we define the partition $g_S$ of $B^{rD}$ into exactly $|B^{|S|}|$ parts of equal size as 
$$
	P_y(g_S) := \{(x_{i_1}, \ldots, x_{i_{|S|}}) = y : y \in B^{|S|}\}.
$$ 
We remark that $H(g_S) = |S|/D$. We shall assign a partition $f_v := g_{S(v)}$ to each vertex $v$; we only need to specify $S(v)$ for all $v$. Denoting $S(X) = \bigcup_{v \in X} S(v)$ for all $X \subseteq V$, we have $H_f(X) = |S(X)|/D$.

$S(v)$ is defined recursively for all trees, level by level. The set $S(R_t)$ for the root of tree $T_t$ ($1 \le t \le r$) is given by
$$
	S(R_t) = \left\{1 + \sum_{s=1}^{t-1} N_s, \ldots, \sum_{s=1}^t N_s \right\}.
$$
We denote $\Sigma := \{1,\ldots,DH\}$. Then, for any non-leaf $v$, the corresponding sets of its children are obtained by adding one element of $\Sigma$ to $S(v)$; all added elements are distinct. That is, for all $u,u' \in c(v)$, we have
\begin{align*}
	S(u) &\subseteq \Sigma\\
	|S(u)| &= |S(v)| + 1\\
	S(u') \cap S(u) &= S(v).
\end{align*}
Let $v$ be a non-leaf on level $k$. Since $|S(v)| = N_t + k$ and $|c(v)| = C_t - k = DH - N_t - k = |\Sigma| - |S(v)|$, we obtain $S(c(v)) = \Sigma$ for all non-leaves $v$.

\begin{lemma} \label{lem:coding_function}
The partitions $f$ form a coding function for $\cl$ with entropy $H$.
\end{lemma}

\begin{proof}
Let us prove that it is indeed a coding function for $\cl$. Since $|S(v)| = N_t + k \le N_t + L = D$ for any $v$ in level $k$ of tree $t$, we obtain $H_f(v) \le 1$ for any $v \in V$. We then need to check that $S(X) = S(\cl(X))$ for any subset $X$ of vertices. We first remark that if $v \in a(u)$, then $S(v) \subseteq S(u)$, hence $S(X) = S(a(X))$. This proves the claim when $\cl(X) = a(X)$. Otherwise, if $c(v) \subseteq a(X)$, then $\Sigma = S(c(v))  \subseteq S(a(X)) \subseteq S(X)$; if $X$ intersects all trees, then $\Sigma = S(\{R_1,\ldots,R_r\}) \subseteq S(X)$. Finally, we have $S(V) = \Sigma$, hence its entropy is equal to $H$.
\end{proof}

\section{Solvability of operators} \label{sec:beyond}

In this section, we justify why we only need to focus on the closure solvability problem. Let us consider the most general way of defining the solvability problem.

\begin{definition} \label{def:coding_function}
Let $V$ be a finite set of $n$ elements, $a: 2^V \to 2^V$, and $A$, $B$ be finite sets ($A$ is referred to as the alphabet, $|A| \ge 2$). A {\em coding function} for $(a,A,B)$ is a tuple $f= (f_1,\ldots,f_n)$ of $n$ partitions of $B$, where each partition is in at most $|A|$ parts, such that $f_{a(X)} = f_X$  for all $X \subseteq V$.
\end{definition}

We say that $a,a': 2^V \to 2^V$ are {\em equivalent} if any tuple of partitions $f$ is a coding function of $a$ if and only if it is a coding function for $a'$.

\begin{theorem} \label{th:cl_equivalent}
Let $a : 2^V \to 2^V$, then there exists a closure operator on $V$ which is equivalent to $a$.
\end{theorem}

\begin{proof}
We take three steps. Firstly, construct the digraph on $2^V$ with arcs $(Y,a(Y))$ for all $Y \subseteq V$. For any $X \subseteq V$, denote the connected component containing $X$ as $C(X)$. Then we claim that $b(X) := \bigcup_{Y \in C(X)} a(Y)$ is equivalent to $a$ (we note that $b$ is extensive). Indeed, if $f$ is a coding function for $a$, then $f_X = f_{a(X)}$. Hence for any $Y \in C(X)$, $f_Y = f_X$ and we obtain $f_{b(X)} = f_X$. Conversely, we have $b(X) = b(a(X))$ and hence if $f$ is a coding function for $b$, then $f_X = f_{b(X)} = f_{b(a(X))} =  f_{a(X)}$ for all $X$.

Secondly, we claim that $c(X) := \bigcup_{Y \subseteq X} b(Y)$ is equivalent to $b$ (we note that $c$ is extensive and isotone). Indeed, if $f$ is a coding function for $b$ and $Y \subseteq X$, then $f_X$ refines $f_Y = f_{b(Y)}$; thus $f_X$ refines $f_{c(X)}$. The converse is immediate hence $f_X = f_{c(X)}$ and $f$ is a coding function for $c$. Conversely, if $f$ is a coding function for $c$, then $f_X = f_{c(X)}$ refines $f_{b(X)}$ and hence is equal to $f_{b(X)}$ for all $X$.

Thirdly, we claim that $\cl(X) := c^n(Y)$ is equivalent to $c$ (we note that $\cl$ is a closure operator). Indeed, if $f$ is a coding function for $c$, then $f_X = f_{c(X)} = \ldots = f_{c^n(X)}$. Conversely, $f_X = f_{\cl(X)}$ refines $f_{c(X)}$.
\end{proof}

\bibliographystyle{IEEEtran}
\bibliography{g}

\end{document}